\documentclass[acmsmall,authorversion=true,nonacm=true]{acmart}

\usepackage{thmtools}
\usepackage{comment}
\usepackage{booktabs} 
\usepackage{multirow}
\usepackage{subfig}
\usepackage{rotating}
\usepackage{tablefootnote}
\usepackage{listings}
\usepackage{enumitem}
\usepackage{centernot}

\usepackage{xargs}  
\usepackage[colorinlistoftodos,textsize=tiny]{todonotes}

\usepackage{forest}
\usepackage{hyperref}
\usepackage{comment}
\usepackage{float}
\usepackage{wasysym}
\ifdefined\Bbbk
\else
\usepackage{amsmath}
\usepackage{amssymb}
\fi

\usepackage[vlined]{algorithm2e} 

\SetAlFnt{\small}
\SetAlCapFnt{\small}
\SetAlCapNameFnt{\small}
\SetAlCapHSkip{0pt}
\IncMargin{-\parindent}
\SetKwProg{Fn}{Function}{}{end}

\usepackage{listings}		
\usepackage{newfloat}
\DeclareFloatingEnvironment{listing}

\usepackage{lipsum}
\usepackage[capitalise]{cleveref}

\newcommand{\mynot}[1]{\neg{#1}}
\newcommand{\myand}{\wedge}
\newcommand{\myor}{\vee}
\newcommand{\union}{\cup}
\newcommand{\mtwosat}{max-2-sat\xspace}
\newcommand{\MTwoSat}{Max-2-sat\xspace}
\newcommand{\vam}{array-aware-matching\xspace}
\newcommand{\VAM}{Array-Aware-Matching\xspace}

\newcommand{\vvv}[1]{\underline{#1}}

\definecolor{mygreen}{rgb}{0,0.6,0}
\definecolor{mygray}{rgb}{0.5,0.5,0.5}

\usepackage{listings}
\input{listings-modelica.cfg}
\setcopyright{none}

\begin{document}
Link to published version: \url{https://dl.acm.org/doi/10.1145/3611661}

\vspace{1cm}

Bibtex entry for citations:
{\scriptsize
\begin{verbatim}
@article{10.1145/3611661,
author = {Fioravanti, Massimo and Cattaneo, Daniele and Terraneo, Federico and Seva, Silvano and Cherubin, Stefano
and Agosta, Giovanni and Casella, Francesco and Leva, Alberto},
title = {Array-Aware Matching: Taming the Complexity of Large-Scale Simulation Models},
year = {2023},
publisher = {Association for Computing Machinery},
address = {New York, NY, USA},
issn = {0098-3500},
url = {https://doi.org/10.1145/3611661},
doi = {10.1145/3611661},
journal = {ACM Trans. Math. Softw.},
month = {jul}
}
\end{verbatim}
}

\clearpage

\title{Array-Aware Matching: Taming the Complexity of Large-Scale Simulation Models}

\author{Massimo Fioravanti}
\email{massimo.fioravanti@mail.polimi.it}
\affiliation{%
  \department{Department of Electronics, Information, and Bioengineering}
  \institution{Politecnico di Milano}
  \country{Italy}%
}
\author{Daniele Cattaneo}
\email{daniele.cattaneo@polimi.it}
\affiliation{%
  \department{Department of Electronics, Information, and Bioengineering}
  \institution{Politecnico di Milano}
  \country{Italy}%
}
\author{Federico Terraneo}
\email{federico.terraneo@polimi.it}
\affiliation{%
  \department{Department of Electronics, Information, and Bioengineering}
  \institution{Politecnico di Milano}
  \country{Italy}%
}
\author{Silvano Seva}
\email{silvano.seva@polimi.it}
\affiliation{%
  \department{Department of Electronics, Information, and Bioengineering}
  \institution{Politecnico di Milano}
  \country{Italy}%
}
\author{Stefano Cherubin}
\email{stefano.cherubin@ntnu.no}
\affiliation{%
  \department{Department of Computer Science}
  \institution{Norges Teknisk-Naturvitenskapelige Universitet}
  \country{Norway}%
}
\author{Giovanni Agosta}
\email{giovanni.agosta@polimi.it}
\affiliation{%
  \department{Department of Electronics, Information, and Bioengineering}
  \institution{Politecnico di Milano}
  \country{Italy}%
}
\author{Francesco Casella}
\email{francesco.casella@polimi.it}
\affiliation{%
  \department{Department of Electronics, Information, and Bioengineering}
  \institution{Politecnico di Milano}
  \country{Italy}%
}
\author{Alberto Leva}
\email{alberto.leva@polimi.it}
\affiliation{%
  \department{Department of Electronics, Information, and Bioengineering}
  \institution{Politecnico di Milano}
  \country{Italy}%
}

\renewcommand{\shortauthors}{Massimo Fioravanti et al.}

\begin{abstract}
Equation-based modelling is a powerful approach to tame the complexity of large-scale simulation problems. Equation-based tools automatically translate models into imperative languages. When confronted with nowadays’ problems, however, well assessed model translation techniques exhibit scalability issues, that are particularly severe when models contain very large arrays. In fact, such models can be made very compact by enclosing equations into looping constructs, but reflecting the same compactness into the translated imperative code is not trivial. In this paper, we face this issue by concentrating on a key step of equations-to-code translation, the equation/variable matching. We first show that an efficient translation of models with (large) arrays needs awareness of their presence, by defining a figure of merit to measure how much the looping constructs are preserved along the translation. We then show that the said figure of merit allows to define an optimal array-aware matching, and as our main result, that the so stated optimal array-aware matching problem is NP-complete. As an additional result, we propose a heuristic algorithm capable of performing array-aware matching in polynomial time. The proposed algorithm can be proficiently used by model translator developers in the implementation of efficient tools for large-scale system simulation.

\end{abstract}


\maketitle
                       
\section{Introduction}
\label{sec:intro}

In modern engineering, dynamic modelling and simulation are ubiquitous~\cite{bib:SchluseEtAl-2018a,bib:BuissonBelloir-2020a}. Besides providing ``virtual prototypes''~\cite{bib:MejiaCarvajal-2017a} to streamline plant and control design activities~\cite{bib:VanBeekEtAl-2014a,bib:VerrietEtAl-2019a}, ``Digital Twins'' -- significantly based on simulation~\cite{bib:AitEtAl-2019a} -- are nowadays the backbone of advanced controls~\cite{bib:AgachiEtAl-2007a,bib:OomenSteinbuch-2020a}, predictive, condition-based and autonomous maintenance~\cite{bib:AivaliotisEtAl-2019a,bib:ChengEtAl-2018a,bib:KhanEtAl-2020a}, anomaly detection, forecast and mitigation~\cite{bib:MarquezEtAl-2018a,bib:TaoEtAl-2018a,bib:HeEtAl-2021a}, continuous integration~\cite{bib:MundEtAl-2018a}, lifelong asset management~\cite{bib:MacchiEtAl-2018a} and many other applications, see for example the survey in~\cite{bib:CiminoEtAl-2019a}. As a result, unprecedented challenges need facing for rapidly creating, modifying and running simulation models of steadily growing size and complexity~\cite{bib:FitzgeraldEtAl-2019a}.

Focusing on simulation models made of Differential and Algebraic Equations (DAE), the \emph{scenario} just sketched has boosted the adoption of declarative, Equation-Based (EB) modelling languages as opposite to procedural, imperative programming ones~\cite{bib:WetterHaugstetter-2006a}. The key feature of EB languages is the ability of separating the activities of writing a model and of producing its solution algorithms. This ability stems from the fact that the fundamental statement in EB languages is the \emph{equation}. Contrary to the \emph{assignment}, where an \emph{l-value} receives the result of computing the expression on the right hand side, an equation just prescribes that the expressions on the left and the right hand side must be made equal -- within convenient tolerances -- at every point in time when the solution of a model is computed during its simulation.

Said otherwise, while assignments directly compose the algorithm to compute the model solution, equations just impose constraints to that solution, therefore saying nothing about the solver (in general, numeric) that will be used to compute it. To synthetically express this separation between describing the model and computing its solution, EB models are called \emph{declarative}.

In synthesis, then, EB languages relieve the analyst from the task of turning equations into imperative code to perform their numerical integration, significantly helping to tame the complexity and rapid evolution of modern simulation problems~\cite{bib:CellierKofman-2006a}. It is the task of a \emph{translator} to automatically turn a declarative model into an equivalent code in some imperative programming language~\cite{bib:FritzsonEtAl-2009a}, which is then fed to a compiler.

This translator-compiler workflow was devised at the outset of EB languages, with the aim of decoupling the generation of imperative code (translation) from its optimised compilation into machine code---a task for which e.g. C compilers are very well suited. However, today's modelling and simulation problems exhibit new characteristics, that require to re-discuss the above translation workflow. A prominent such characteristic, on which we focus in this paper, is the presence of large (and possibly multi-dimensional) array variables and equations. This feature is distinctive of ``large-scale'' models. Think for example of a 3D thermal model for a solid body with fine-grained spatial discretisation: the model will contain energy dynamic balance equations for each of the many subvolumes into which the solid will be partitioned, and these equations will account for thermal exchanges with the adjacent volumes. It is quite natural to write such a model compactly in EB form by defining suitable array variables and equations, in the latter case by means of looping constructs (examples follow starting from~\cref{sec:defs-motivation}).

In such cases, as we will show, a trade-off is easily observed. Writing the model directly as imperative code is far more complex, error prone and hard to maintain than adopting an EB declarative framework and obtaining the imperative code by automatic translation. But on the other hand, the code obtained by automatic EB-to-imperative translation is significantly less efficient than the one manually written as imperative.

We argue that the origin of this inefficiency mainly resides in the way EB translators manage array variables and equations. Current production-grade EB translators just treat each component of an array variable or equation as an individual scalar one, which results in a loss of structural information that imperative language compilers cannot efficiently recover~\cite{bib:AgostaEtAl-2019b}. As such, we argue that to improve both the translation time and the efficiency of the produced imperative code, it is necessary to make the translation process ``array-aware''.

In this paper we offer a contribution to this end, aiming \emph{both} for an efficient translation and an efficient imperative code. In detail,

\begin{enumerate}

\item we define a figure of merit to quantify how much the looping constructs that make an equation model
      compact carry over to its imperative translation; building on this figure of merit, we consequently
      define as \emph{optimal} an array-aware matching that maximally preserves the said looping constructs;

\item we prove that the optimal array-aware matching problem is NP-complete -- contrary to scalar matching, 
      which can be solved in polynomial time and where no such optimality makes sense;

\item we propose a heuristic algorithm to approximate optimal array-aware matching in polynomial time.

\end{enumerate}

To the best of our knowledge, we are the first to introduce an idea of matching optimality tied to the efficiency of the obtained imperative code, as well as to propose a heuristics that aims for that optimality besides for a fast matching process.

\paragraph{Organisation of the paper}

The rest of the paper is organised as follows.
In \cref{sec:bg-translation} we briefly introduce some definitions and the architecture of state-of-the-art approaches.
Then, we discuss the history of EB model translators and other array-aware approaches in \cref{sec:related-work}.
We delineate the scope and purpose of our contribution in \cref{sec:defs-motivation}.
In \cref{sec:graphs-pb-statement}, we define \emph{array-aware matching} and its optimality metric,
while we prove its NP-completeness in \cref{sec:complexity-proof}.
Finally, in \cref{sec:algos}, we show an approximate algorithm for array-aware matching, 
and in \cref{sec:conclusions} we draw some conclusions and highlight future research directions.

\section{Background on model translation}
\label{sec:bg-translation}

In this section we outline the foundations of automatic EB-to-imperative translation and compilation. To avoid confusion between scalar and array problems, we first provide a few definitions.

\begin{definition}[scalar variable]
A \emph{scalar variable} is an instance of a system property, identified by a name, whose value at every instant is fully defined by a scalar number and -- possibly -- a unit of measurement.
\end{definition}

\begin{definition}[array variable]
An \emph{array variable} is a collection, identified by a name, of one or more scalar variables, each one referenceable through one or more integer indices.
\end{definition}

\begin{definition}[array equation]
\label{def:arrayequation}
An \emph{array equation} is a collection of one or more scalar equations, expressed compactly as a single parametric equation that references by index one or more scalar components of one or more array variables.
\end{definition}

\begin{definition}[array dimensionality and size vector]
The \emph{dimensionality} of an array is the number of dimensions of that array, that is, the number of integer indices needed to reference a single scalar component in it. We assume by convention that the said indices are 1-based. Their maximum values collectively form the array \emph{size vector}.
\end{definition}

It follows that, in EB modelling languages, \emph{array variable}s are akin to the concept of ordinary multidimensional arrays, commonly found in many programming languages. Also, in EB modelling languages, an \emph{array equation} is obtained by encasing a scalar one in one or more nested looping constructs, that define the indices of the contained scalar variables within the arrays of which they are part, as well as the ranges for the said indices.

It is important to highlight that the mentioned looping constructs -- differently from those of imperative programming languages -- allows the user to predicate on array equations, such as
\begin{equation}
 x[i]=y[i]\; \forall i \in [1,3];
 \label{eqn:bg-rolled}
\end{equation}
to be intended as an abbreviation of
\begin{equation}
 x[1]=y[1];\;x[2]=y[2];\;x[3]=y[3];
 \label{eqn:bg-unrolled}
\end{equation}

This noted, the translation-compilation process -- as per the current state of the art in both research and production tools, see e.g.~\cite{bib:CellierKofman-2006a,bib:FritzsonEtAl-2019a} and~\cite{bib:Dymola,bib:JModelica,bib:OpenModelica,bib:SimulationX} respectively -- can be divided into the following steps.

\begin{description}

\item[Flattening] The model equations, independently of the way they were input by the user -- e.g. as a single set, by hierarchically instantiating and interconnecting subsystems, or anyhow else -- are brought to be one set of scalar DAEs. This step contains a sub-step named \emph{loop unrolling}, in which each  expression in the form shown in \cref{eqn:bg-rolled} is replaced by its set of scalar components, as shown in \cref{eqn:bg-unrolled}. The outcome of flattening is thus a DAE system with scalar equations and variables.

\item[Matching] This step (hereinafter denoted as \emph{scalar} matching when confusion may occur) consists of coupling each (scalar) equation to one (scalar) variable, meaning that the equation is the initial candidate for computing the variable at simulation time. A failure in matching indicates a model inconsistency (e.g. and most typically, an equations/variables imbalance).
In the general case, matching may also require an index reduction sub-process, implemented by methods such as the Pantelides algorithm~\cite{bib:Pantelides-1988a}. Our paper does not address index reduction.

\item[Scheduling] This step determines the order in which the equations are solved.
The (scalar) equations of the system are ordered accordingly to their mutual dependencies, as established by the matching process.
For example, the equation matched with variable $v$ -- that is, the candidate one to compute $v$ -- is scheduled before all other equations in which $v$ appears.
The ideal result would allow to compute the solution variable by variable, in sequence. This is hardly ever obtained, however.
During the scheduling step, some cyclic dependency among variables may arise, and as a result the one-to-one relationship established by the matching process between those variables and their candidate equations cannot be maintained.
Cyclic dependencies indicate the existence of a so-called \emph{Strongly Connected Component} (SCC): the involved variables will need computing all together, most often numerically.
The Tarjan algorithm~\cite{bib:Tarjan-1972a} is a commonly used means to determine the equation solution order and to identify SCCs.

\item[Code generation] The last step is the generation of imperative simulation code for a specific choice of the numerical integration algorithm. The simulation code can be self-contained for certain numerical integration algorithms, such as explicit ones, or can rely on external solver libraries, such as those in the SUNDIALS suite~\cite{bib:HindmarshEtAl-2005a}.

\end{description}

After the translation and compilation process is complete, the obtained executable code is run to produce the simulation output, in the form of a table with the value of the model variables as a function of simulation time.
This process is often part of a graphical modelling environment for rapid prototyping. Once the modeller decides to perform a simulation, all the translation, compilation and execution steps are on the critical path toward getting the simulation results. Consequently, the shortening of compilation and simulation times is especially important for this kind of use-cases.

\section{Related work}
\label{sec:related-work}

In this section we spend some words to relate our proposal to the research \emph{scenario} on EB modelling in general, and to neighbouring research on EB-to-imperative translation in particular.
 
In the landscape of simulation languages, EB ones appeared and gained visibility in the 80s/90s of the last century; notable examples are Omola~\cite{bib:Andersson-1989a}, Dymola~\cite{bib:Elmqvist-1979a} and gPROMS~\cite{bib:BartonPantelides-1994a}. A taxonomy would stray from the scope of this paper; the reader interested in the historical \emph{panorama} can refer e.g. to~\cite{bib:Cellier-1983a}. Worth noticing, however, is the common ancestor for the boost of the declarative approach stemming from studies such as~\cite{bib:AstromKreutzer-1986a} and~\cite{bib:ElmqvistMattsson-i989a}, where the idea that simulation languages had to abandon the imperative setting was set forth and preliminarily exploited. Research on the matter thus focused at first on the model manipulation~\cite{bib:HuntCremer-1997a,bib:MattsonEtAl-1997a} required by going declarative, and after a long systematisation process, this resulted in the birth of the Modelica language~\cite{bib:MattssonEtAl-1998a} to which we refer herein (though all the ideas we propose are general to the OOM context).

The engineering use of EB languages and tools sustainably spread out in various domains, ranging
from the chemical~\cite{bib:AsteasuainEtAl-2001a}
and process industry~\cite{bib:GarciaEtAl-2002a}
to power generation~\cite{bib:CasellaLeva-2006a}
and transmission~\cite{bib:SusukiHikihara-2002a},
mechatronics~\cite{bib:VanAmerongenBreedveld-2003a}
and robotics~\cite{bib:HirzingerEtAl-2005a},
automotive~\cite{bib:JairamEtAl-2008a},
and vehicles at large~\cite{bib:DonidaEtAl-2008a},
aerospace~\cite{bib:MoormannLooye-2002a},
buildings~\cite{bib:SowellHaves-2001a}
and more,
including control design~\cite{bib:CasellaEtAl-2008a}
and diagnostics~\cite{bib:BunusLunde-2008a};
the papers in the necessarily limited list above also contain interesting bibliographies for the reader willing to investigate further.

Together with testifying the success of the EB approach, however, the expansion just mentioned also shed light on some relevant limitations of the existing EB tools~\cite{bib:SahlinEtAl-2003a} -- not of EB languages by themselves, it is worth stressing -- especially when dealing with large-size models~\cite{bib:LinkEtAl-2009a,bib:CasellaEtAl-2017a}. This was the motivation for a first wave of tool optimisation, having as a major point the introduction of sparse solvers, a well-treated and long-lasting matter in domain-specific tools -- see e.g.~\cite{bib:Friedman-1992a,bib:StadtherrWood-1984a,bib:StadtherrWood-1984b,bib:Li-2010a} -- but a source of challenges in the inherently multi-domain EB one~\cite{bib:Paloschi-1996a,bib:WangKolditz-2010a}. Examples of this research -- with specific reference to Modelica given our scope -- are~\cite{bib:SandholmEtAl-2006a,bib:Casella-2015a}.

The possibility of \emph{solving} large models fast enough to widen the EB applicability perimeter evidenced however a second type of tool limitation, concerning the \emph{translation} rather than the solution of such models~\cite{bib:Sjolund-2015a,bib:SchweigerEtAl-2020a}. The matter became critical in recent years, together with the emergence of problems that require model-based prototyping~\cite{bib:MaloneEtAl-2016a} and can scale up to the order of $10^5$ equations. When such models become part of the inherently iterative engineering process, the time spent in translating and compiling them can be comparable to that spent in running simulations, if not even dominant~\cite{bib:BaldinoThesis-2018}. For the sake of clarity it is worth noticing that the million equation barrier was already approached in the past~\cite{bib:TezduyarEtAl-1996a} and in some domains nowadays well trespassed~\cite{bib:RouetEtAl-2020a} by simulation tools, but these tools \emph{are not of the EB type}, and most notably, do not separate model description and solution---which is a primary goal of the EB approach.

As a result, EB tools are nowadays undergoing a second wave of optimisation, directed to efficient translation.
Open-source parsers~\cite{pop2019frontend} for EB languages such as Modelica are available providing some degree of array preservation, thereby enabling the research community to experiment with making the translation pipeline array-aware.
In this relatively new effort, a primary objective is to achieve an $\mathrm{O}(1)$ scaling of the translation time with the size of the model arrays, that as already noted are the main cause for the inefficiency of scalar-only model manipulation. In this context, the nearest neighbouring work to our research is the paper by Zimmermann \emph{et al.}~\cite{bib:ZimmermannEtAl-2019a}, who introduce the concept of ``set-based graph'' as a means to re-state the matching problem (originally scalar) in such a way to achieve an $\mathrm{O}(1)$ translation, together with proposing algorithms for other manipulation steps related to matching, such as the management of strongly connected components and scheduling.

The main difference of our research with respect to~\cite{bib:ZimmermannEtAl-2019a} is a twofold instead of a single goal. More precisely, we do not aim just for an efficient translation, but also for an efficient simulation code. If the efficiency of the simulation code is taken into account, 
the number of looping constructs in the EB model that are preserved in the imperative one comes to matter a lot. The set-based approach of \cite{bib:ZimmermannEtAl-2019a} is not designed to take this aspect into account.
Aiming at loop preservation straightforwardly entails the introduction of an idea of optimality. This moves the focus from array-aware matching to \emph{optimal} array-aware matching, and owing to its NP completeness, to the need for heuristics.

Other works have addressed the \vam problem, such as~\cite{bib:SchuchartEtAl-2015} which correctly noticed that preserving looping constructs can positively impact both translation and simulation time. They also present a prototype translator that is limited to handling systems without algebraic loops and resorts to flattening the model completely in cases where their \vam algorithm does not produce a solution.
The paper \cite{otter2017transformation} addresses array-aware index reduction, and the corresponding Modia implementation~\cite{bib:ModiaBase} also includes an \vam algorithm that however is very simple and cannot split variable and equation nodes, thus requiring the modeller to pre-process the input model so as to make sure that every array variable can be matched to exactly one array equation.
Compared to the two previously quoted works, our paper introduces the concept of optimal \vam and proves that achieving such optimality is an NP complete problem, as well as presenting a more complete \vam algorithm.

\section{Research motivation}
\label{sec:defs-motivation}

In this section, we discuss the inefficiencies that arise from a non array-aware model translation.
To ground the discussion on an example, consider the model of a thermally insulated metal wire with prescribed temperatures at its ends. The evolution of this system is ruled by the one-dimensional Fourier equation. Carrying out a uniform spatial discretisation with the finite-volume approach results in the following system of differential equations:
\begin{equation}
\label{eq:wire}
c \dot{T}_i = \begin{cases}
                g  (2  T_{\textit{left}} - 3  T_i + T_{i+1})  & i = 1          \\
                g  (T_{i-1} - 2  T_i + T_{i+1})           & \forall \; i \in [2, N-1] \\
                g  (T_{i-1} - 3  T_i + 2  T_{\textit{right}}) & i = N
              \end{cases}
\end{equation}
where $N>3$ is the number of finite volumes, $c$ is the thermal capacity of a volume, $g$ the thermal conductance between the centres of two adjacent volumes, $T_i$ is the temperature of volume $i$, with $T_1$ being the leftmost and $T_N$ being the rightmost volume, and finally $T_{\textit{left}}$ and $T_{\textit{right}}$ are the prescribed side temperatures.

When expressed in Modelica, the wire model reads as follows.

\begin{listing}[h!]
\begin{tabular}{c}
\begin{lstlisting}[language=Modelica]
model Thermal1D
    parameter Integer N      = 5;
    parameter Real    g      = 0.00314785;   // W/K
    parameter Real    c      = 0.2707936;    // J/K
    parameter Real    Tleft  = 400 + 273.15; // K
    parameter Real    Tright = 20 + 273.15;  // K
    Real T[N](each start=Tright);                   // array variable

equation
    c*der(T[1])  = g*(2*Tleft - 3*T[1] + T[2]);
    for i in 2:N-1 loop                             // looping construct to express
        c*der(T[i]) = g*(T[i-1] - 2*T[i] + T[i+1]); // an array equation
    end for;
    c*der(T[N]) = g*(T[N-1] - 3*T[N] + 2*Tright);
end Thermal1D;
\end{lstlisting}
\end{tabular}
\caption{\label{lst:wire}The model of a wire, as shown in \cref{eq:wire}, as expressed in Modelica code.}
\end{listing}
The remarkable similarity between the Modelica model and original system of equations is apparent; the \emph{der} operator is used for expressing the derivative with time.

It is evident that the availability of array equations make EB models assume a compact and easily readable form. However, present state-of-the-art translators provide looping constructs \emph{only} as a convenience for the modeller, and do not take advantage of this structural information to improve the translation efficiency, nor that of the generated simulation code. As a consequence, some inefficiencies arise that would be completely unexpected in the world of imperative programming languages.

A first such inefficiency is that the amount of both time and memory needed for the translation scale superlinearly with the size of array equations, rather than exhibiting the expected $\mathrm{O}(1)$ scaling typical of imperative languages. In fact, the compilation of imperative programming languages treats looping constructs explicitly. The instructions inside a loop are represented in the compiler data structures just once, irrespective of the loop iteration count, and all subsequent compilation steps are performed on this compact data structure.

On the contrary, state-of-the-art EB translation algorithms are designed to only work in terms of scalar equations (see \emph{Flattening} in \cref{sec:bg-translation}), which are stored in the translator data structures individually. These scalar equations are then passed on to subsequent translation steps, some of which scale superlinearly. It is not uncommon for models exceeding the $10^5$ equations mark to require translation times in the order of hours, and working memory in the order of hundreds of gigabytes~\cite{bib:BaldinoThesis-2018}.

A second and consequent inefficiency is that the size of the produced imperative code scales linearly with the size of array variables and equations.
Since the original EB model is scalarised, the simulation consists of procedures containing repetitive code instead of looping constructs, often amounting to several gigabytes for large-scale models.
Although imperative compilers are very efficient and can achieve $\mathrm{O}(1)$ scaling in compilation with respect to array sizes, this is only possible if they are given a source code with loops, not long lists of repetitive statements.

The last inefficiency regards the performance of the machine code produced by the compiler when fed with the automatically generated imperative code.
Modern computer architectures are built upon assumptions such as the \emph{locality principle} for both data and code, which are the theoretical foundations for caches~\cite{bib:Drepper-2007a} and other microarchitectural optimisations.
However, the code produced by current-generation model translators is not able to exploit these optimisations.
First, the execution of large blocks of straight-line code requires frequent \emph{instruction} cache invalidation.
Furthermore, the loss of structural information about arrays leads to the generation of code that exhibits irregular data access patterns.
Therefore, repetitive machine code is not only large, but runs significantly slower than equivalent hand-written code, in some cases $100$ times or more~\cite{bib:AgostaEtAl-2019a}.

Both translators and compilers can in principle infer looping constructs and improve access locality, but their ability to perform this kind of optimisation is limited~\cite{bib:LarsenAmarasinghe-2000a}.
Additionally, such inferences require to repetitively scan the list of flattened statements, making $\mathrm{O}(1)$ scaling apparently impossible. 
We thus argue that a far better approach would be to preserve array-awareness throughout the translation process, rather than try to regain it \emph{a posteriori}.

Summing up, all the presented inefficiencies share their root cause in the fact that existing model translation algorithms are not array-aware.
In this paper, we begin an effort to fill this gap.
Since the flattening step is trivial to extend for array-awareness -- it suffices to not unroll looping constructs -- we focus on the second and first truly key step, i.e., the matching problem.

\section{Graph representation and problem statement}
\label{sec:graphs-pb-statement}

In this section, we introduce the notation to describe array-aware algorithms, and present the formal statement of the array-aware matching, the algorithm we focus on in this paper.

We denote with $G=(N,D)$ a generic graph having set of nodes $N$ and set of arcs $D$.
Any set of scalar equations $E$ in the scalar variables $V$ can be represented with a bipartite graph in which
\begin{itemize}
    \item $N = V \cup E$,
    \item $V \cap E = \emptyset$,
    \item $D \subseteq V \times E$,
    \item the presence of an arc $(v_i, e_j)$ indicates that variable $v_i$ appears in equation $e_j$.
\end{itemize}

In this context, the set of arcs $D$ represents dependencies between variable and equation nodes.

To include array variables and equations -- including multidimensional ones -- we extend the above notation as follows.
We indicate with $\underline{v}_i$ the generic multidimensional array variable, whose scalar components in the physical problems we target are invariantly real numbers.
We denote by $\zeta_i$ the dimensionality of $\underline{v}_i$.
We denote with $K=\{k_1, \dots, k_\delta, \dots, k_{\zeta_i}\}$ the sequence of positive integers needed to reference a scalar component within the array variable $\underline{v}_i$.
Each generic index $k_\delta$ ranges from $1$ to the size of the corresponding dimension $\delta$, with $1 \leq \delta \leq \zeta_i$.
For compactness, we synthetically write $v_{i,K}$ to indicate the scalar component in $\underline{v}_i$ of indices $\{k_1 \dots k_{\zeta_i}\}$.

Consistently with the above notation, we define the following.
\begin{itemize}
    \item $\underline{V}$ is the set of array variables $\{\underline{v}_i\}$, $i = 1, \dots, |\underline{V}|$ in the model; if the model contains scalar variables, these are considered array variables of unitary dimensionality and size;

    \item $\underline{E}$ is the set of array equations $\{\underline{e}_j\}$, $j = 1, \dots, |\underline{E}|$ in the model; if the model contains scalar equations, they are considered array equations of unitary dimensionality and size.
\end{itemize}

A model containing both scalar and array variables and equations can thus be represented with any of the two equivalent bipartite graphs defined as follows.

\begin{enumerate}
    \item The first one is obtained by just setting $N = \{v_{i,K}\} \cup \{e_{j,L}\}$, and $D \subseteq V \times E$
          as the scalar dependencies. The presence of an arc
          $(v_{i,K}, e_{j,L})$ indicates that the scalar variable $v_{i,K}$ appears in the scalar equation $e_{j,L}$.
          Observe that in the topology of such a graph any information concerning the existence of array variables and
          equations is lost. We name this the \emph{flattened} graph.

    \item The second one is obtained by setting $\underline{N} = \underline{V} \cup \underline{E}$, and $\underline{D}
          \subseteq \underline{V} \times \underline{E}$ as the array dependencies. In this case the presence of
          an arc $(\underline{v}_i, \underline{e}_j)$ indicates that at least one scalar variable in $\underline{v}_i$
          appears in at least one scalar equation in $\underline{e}_j$. We name this the \emph{array} graph.
\end{enumerate}

\begin{figure}[tb]%
\begin{tabular}{c}
\begin{lstlisting}[language=Modelica]
c*der(T[1])  = g*(2*Tleft - 3*T[1] + T[2]);     // e1
for i in 2:N-1 loop
    c*der(T[i]) = g*(T[i-1] - 2*T[i] + T[i+1]); // e2
end for;
c*der(T[N]) = g*(T[N-1] - 3*T[N] + 2*Tright);   // e3
\end{lstlisting}
\end{tabular}
\hspace*{0.5em}
$\vcenter{\hbox{\includegraphics[width=0.4\textwidth]{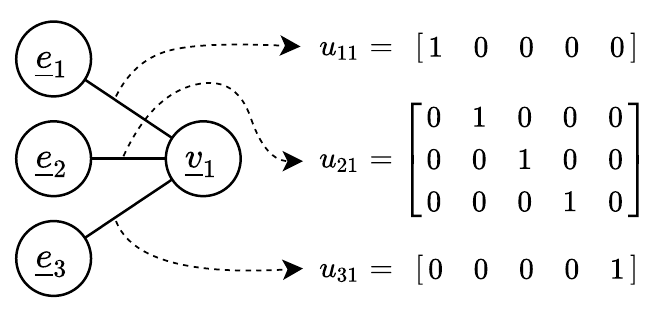}}}$
\caption{\label{lst:forloopmodelica}The array equations of the Modelica model of a wire previously shown in \cref{lst:wire}, and its corresponding array graph. Variable $\underline{v}_1$ is $\dot{T}$, the only non-state array variable in the model. Thus the columns of the $u$ matrices correspond to $\dot{T}_1$, $\dot{T}_2$, $\dot{T}_3$ and so on.}
\end{figure}

The above definitions imply that the array graph is homomorphic to the flattened one.
However, for the two graphs to be equivalent,
each arc in the array one needs to carry information about which components of the connected array variable and equation it refers to --- a matter that does not pertain to the scalar case.
Therefore, the arcs of the array graph must be endowed with the information needed to reconstruct the arcs of the flattened graph.
To formalise this, we introduce the concept of \textit{local multidimensional incidence matrix}.
Given an array equation $\underline{e}_i$ and an array variable $\underline{v}_j$ of dimensionality $\zeta_i$ and $\zeta_j$ respectively, let $K_i=\{k_{i,1},...k_{i,{\zeta_i}}\}$ be the sequence of indices for $\underline{e}_i$, and $L_j=\{l_{j,1},...l_{j,{\zeta_j}}\}$ be the sequence of indices for $\underline{v}_j$.
The multidimensional local incidence matrix $u_{i,j}$ has dimensionality $\zeta_{ij}=\zeta_i+\zeta_j$, and its sequence of indices $Q_{ij}$ is the concatenation of $K_i$ and $L_j$.
Its generic element is $1$ iff the scalar variable $v_{j,K}$ appears in the scalar equation $e_{i,L}$, else it is $0$. We name $U$ the set of local multidimensional incidence matrices.
An example of a Modelica model and the corresponding array graph can be found in~\cref{lst:forloopmodelica}.

The notation we just provided is capable of representing arbitrary multidimensional incidence matrices. However, it should be noted that those produced by EB models coming from equations of physics are significantly structured, and present patterns that arise out of the looping constructs and expressions used for accessing array variables within array equations. It follows that, although from a theoretical standpoint it is convenient to reason in terms of multidimensional incidence matrices, an industry-grade implementation would rely on an efficient pattern-based data structure to achieve $\mathrm{O}(1)$ scaling. This important matter is discussed in \cref{sec:gmis-vaf}.

\subsection{The \vam problem}
\label{sec:arrayawarematching}

The \vam is an operation that takes as input an array graph $\underline{G} = ( \underline{V} \union \underline{E}, \underline{D} )$, where we recall that $\underline{D}$ is the set of dependencies between array variables $\underline{V}$ and array equations $\underline{E}$. Additionally, every dependence $\underline{D}$ has an associated local multidimensional incidence matrix $u_{i,j}$.
The \vam produces as output an array graph $\underline{G}' = ( \underline{V} \union \underline{E}, \underline{D}' )$ with the following properties.
\begin{enumerate}
\item $\underline{D}' \subseteq \underline{D}$
\item Each arc $( \underline{v}_{i}, \underline{e}_{j} ) \in \underline{D}'$ has an associated local multidimensional incidence matrix $m_{i,j}$, where $m_{i,j} ( \underline{v}_{i,K}, \underline{e}_{j,L}) \in \{0,1\}$. We call $M$ the set of $m_{i,j}$.
\item $\forall m_{i,j} \in M$, $m_{i,j}$ has the same size and dimensionality of $u_{i,j}\in U $
\item $\nexists ( \underline{v}_{i}, \underline{e}_{j} ) \in \underline{D}' : m_{i,j}$ is a zero matrix
\item $u_{i,j} ( \underline{v}_{i,K}, \underline{e}_{j,L} ) = 0 \implies m_{i,j} ( \underline{v}_{i,K}, \underline{e}_{j,L} ) = 0$
\item $\forall \underline{v}_{i,K} \in \underline{V} : \exists ! m_{i,j} ( \underline{v}_{i,K}, \underline{e}_{j,L} ) = 1$
\end{enumerate}

Property (1) states that no new arcs are added to the graph, (2--3) denote as ``matching matrices'' $m$ the local incidence matrices of the output graph, (4) states that -- consistently -- arcs fully unmatched are removed, (5) ensures that a scalar equation is matched only with a scalar variable it contains, and (6) that each scalar equation is matched to one and only one scalar variable.

\begin{definition}[optimal array-aware matching]
\label{def:optimality}
Let us define the \emph{expansion function} $f_e(m)$, which maps a local incidence matrix $m$ to the number of looping constructs required to implement the matching it expresses.
$f_e(m)$ has a value of zero if and only if $m$ is a zero matrix, otherwise $f_e(m) \geq 1$.
We say that a matching $(\underline{G}', M)$ is optimal if it minimizes the following:
\begin{equation}
\Omega(\underline{G}, M) = \sum_{m \in M}{f_e(m)}
\end{equation}
\end{definition}

For the purpose of this work, it is not necessary to fully specify the expansion function, because different model translators may be able to ``efficiently'' handle only a subset of all possible matching matrices.
In this discussion, the metric of ``efficiency'' is therefore the ability to represent the matching described by $m_{i,j}$ as a single looping construct consisting of the scalar equations of the node adjacent to the corresponding arc $d_{i,j}$.
Without loss of generality we can consider the ideal case in which every matching matrix maintains its correspondence with exactly one looping construct.
In this case, the $f_e$ function takes this form:
\begin{equation}
f_e(m) = \left\{
\begin{array}{ll}
0 & \mbox{iff $m$ is a zero matrix}\\
1 & \mbox{otherwise}\\
\end{array}
\right.
\end{equation}

In this case, the optimality metric $\Omega$ is equal to the number of arcs where at least one equation is matched.
Intuitively, the more arcs have a matching, the higher the number of times the same array equation has been matched with multiple variables, and for each variable matched with the same array equation a separate looping construct must be introduced.
Therefore, minimizing $\Omega$ means minimizing the number of looping constructs.

\subsubsection{Just enough, but not too much}
It is important to stress that the optimality metric $\Omega$ is but a proxy for identifying the matching that best improves both translation and simulation time. With the definition we gave of array equation in this paper (Definition~\ref{def:arrayequation} in Section~\ref{sec:bg-translation}) where individual scalar equations in an array equation differ only by the array indices, reducing the number of looping constructs does indeed result in more efficient code. Other works however, such as \cite{otter2017transformation} and the corresponding Modia implementation \cite{bib:ModiaBase} allow if statements in array equations thereby allowing to merge in a single array equation also scalar equations that are structurally different. For example, the wire model written with such an extended definition of array equations would read as
\begin{lstlisting}[language=Modelica]
for i in 1:N loop
    c*der(T[i]) =      if(i == 1) then g*(2*Tleft - 3*T[1] + T[2])
                  else if(i == N) then g*(T[N-1] - 3*T[N] + 2*Tright)
                  else                 g*(T[i-1] - 2*T[i] + T[i+1]);
end for;
\end{lstlisting}
Although it would appear that such an array equation could improve the optimality metric by requiring a single looping construct that handles all the wire finite volumes instead of requiring one looping construct plus two scalar equations for the first and last volume, this would \emph{not} result in efficient code. Indeed, if such an array equation were brought as-is till code generation, the imperative code would need if statements in the loop to handle the differences in the equation structure, and the cost at run-time of the if statements would need to be paid multiplied by every loop iteration and additionally multiplied by every simulated time step. This is a strong point in favour of our definition of array equations, that explicitly disallows grouping structurally different scalar equations in a single array equation.
Moreover, to achieve the best efficiency, a translator would need to perform symbolic manipulations to transform inefficient code such as the one above by moving the structurally dissimilar equations outside of the loop before the matching step, should the modeller decide to write the model in that form.

\section{Complexity of optimal \VAM}
\label{sec:complexity-proof}

In this section we show that for what concerns the matching problem, preserving the array structure of both equations and variables results in NP-completeness. This theoretical result motivates the need for heuristic algorithms, that will
be presented in \cref{sec:algos}.

\begin{theorem}[Complexity of optimal \vam]
The problem of producing an optimal \vam is NP-complete.
\label{thm:vamnp}
\end{theorem}

\begin{proof}
We prove \cref{thm:vamnp} by reducing the \mtwosat problem to optimal \vam. \MTwoSat was proven NP-complete by Garey, Johnson and Stockmeyer~\cite{bib:GareyEtAl-1976a}, and reads as follows: given a Boolean formula in conjunctive normal form, where each clause contains at most two literals, find a literal assignment such that the maximum number of clauses is satisfied.

First of all, we establish a procedure that allows to represent a \mtwosat problem in terms of \vam. To this end we introduce three formal modifications to the way \mtwosat is expressed in order to simplify the reduction process.

\begin{enumerate}

\item Every clause in the form $(a)$ is rewritten as $(a \vee a)$. It is evident that this rewriting
      does not change the value of that clause for any literal assignment.

\item Instead of considering the input to \mtwosat as a conjunction of clauses, we consider it as an ordered list.
      This is legitimate, as maximising the number of satisfied clauses does not require knowing if the entire formula
      is satisfied.

\item We assume that, when traversing the list of clauses, the first encountered occurrence of each literal is not
      negated. Should this be false, one would simply have to replace that literal with another one defined as its
      complement.

\end{enumerate}

The above said, to encode a \mtwosat instance into a \vam one, we start by recalling that
\begin{equation*}
(a \myor b) = (a \myand \mynot{b}) \myor (\mynot{a} \myand b) \myor (a \myand b)
\end{equation*}
and we replace each OR clause in the list with the above defined equivalent triple of AND ones. Observe that only one of the three AND clauses can be true, hence maximising the number of satisfied AND clauses in the reformulated list is equivalent to the original problem.

\begin{figure}
\begin{center}
\includegraphics{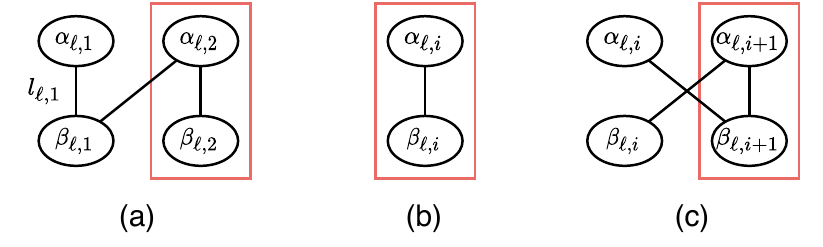}
\caption{\label{fig:butterflies}Subgraphs used by the procedure for building the intermediate graph $G$ from a list of AND clauses.}
\end{center}
\end{figure}

Let us start with two empty sets of nodes $A = \emptyset$, $B = \emptyset$.
Then, we build an intermediate flattened bipartite graph $G = (A \union B, L)$ by traversing the list of AND clauses.
By construction, we will retain the invariant $A \cap B = \emptyset$.
For each clause $c$ we operate as follows.

First, we consider the first literal $\ell$ in the clause.
If this is the first occurrence of $\ell$, we add to $G$ the subgraph shown in \cref{fig:butterflies} (a).
In doing so, $\alpha$ nodes should be considered belonging to the $A$ set while $\beta$ nodes should be considered part of the $B$ set.
We name $e_{\ell,1}$ the first edge associated to literal $\ell$.
We name $\alpha_{\ell,1}$ and $\beta_{\ell,2}$, respectively the \emph{start} and the \emph{end node} for that literal.

If $\ell$ was already encountered and appears here in non-negated form, we add to $G$ the subgraph shown in \cref{fig:butterflies} (b), and we connect $\alpha_{\ell,i}$ to the end node of the literal.
The start node of the literal does not change, while its end node becomes $\beta_{\ell,i}$.

If $\ell$ was already encountered and appears here in negated form, we add to $G$ the subgraph shown in \cref{fig:butterflies} (c), and we connect $\alpha_{\ell,i}$ to the end node of the literal.
The start node of the literal does not change, while its end node becomes $\beta_{\ell,i}$.

We repeat the subgraph insertion process for the second literal in the clause.

For each clause $c$, we define a set of four nodes -- taken from the two subgraphs just added -- as the union of nodes highlighted in red in \cref{fig:butterflies}. We name this set \emph{clause nodes} of $c$ and we name it $N_c$.

We repeat the above for all clauses.
When the end of the list is reached, we connect the end node of each literal with its start one. Doing so, we create for each literal a simple cycle $\alpha_{\ell,1}$, $\beta_{\ell,1}$, $\ldots$, $\alpha_{\ell,n_\ell}$, $\alpha_{\ell,1}$ of even cardinality $n_\ell$, ordered as just indicated. We number the edges of each cycle as $l_{\ell,1}$ through $l_{\ell,n_\ell}$. This concludes the construction of graph $G$.

Graph $G$ is a flattened graph and as such does not have the form required by \vam, thus we need to
construct a different graph $\underline{G}=(\underline{E} \union \underline{V}, \underline{D})$ homomorphic to $G$, and expressed in terms of array equations and variables. Therefore, $\underline{G}$ will
satisfy the following properties:

\begin{itemize}
\item For each node $\alpha_{\ell,i} \in A$ not part of a clause node $N_c$, there exists a variable $\underline{v}_{n} \in \underline{V}$ with dimensionality 1 and size 1.
\item For each node $\beta_{\ell,j} \in B$ not part of a clause node $N_c$, there exists an equation $\underline{e}_{m} \in \underline{E}$ with dimensionality 1 and size 1.
\item For each clause node $N_c$ there exists a variable $\underline{v}_{n} \in \underline{V}$ and an equation $\underline{e}_{m} \in \underline{E}$, both with dimensionality 1 and size 2.
\item The local incidence matrix $u_{n,m}$ is a square identity matrix in order to represent the original relationship found in the scalar graph.
\item No other variables nor equations exist in $\underline{G}$.
\item Dependencies arcs $\underline{D}$ are constructed so as to make $\underline{G}$ homomorphic to $G$.
\end{itemize}

In other words, all the $\alpha$ nodes are considered variables, all the $\beta$ nodes are considered equations, and each labelled set of four nodes $N_c$ forms an array equation of size $2$, and a corresponding array variable of size $2$, while unlabelled nodes translate to scalars.

To carry on, we now need the following lemma.

\begin{lemma}[Complexity of the construction of $G$]\label{lemma:match-graph}
The graphs $G$ and $\underline{G}$ can be constructed in polynomial time.
\end{lemma}

\begin{proof}
The construction process which defines $G$ is linear with the number of clauses and the creation of $\underline{G}$ can be done by simply enumerating the nodes and edges of $G$.
\end{proof}

Back to the main proof, once we obtained the matched graph by executing \vam on $\underline{G}$, we assign to each literal $\ell$ the boolean value \emph{true} if arc $l_{\ell,1}$ has been selected for the matching, \emph{false} otherwise. This last step can also be performed in polynomial time.

Given the definition of optimal \vam of~\cref{def:optimality}, an algorithm capable of solving the problem will maximise the number of $(l_{\ell_1,i},l_{\ell_2,j})$ arc pairs between nodes belonging to the clause node sets $N_c$. In fact, each pair contributes only an unitary weight to the optimality metric $\Omega$, while two non-paired arcs will contribute a weight of two.

By construction of the bipartite graph $G$, each node, be it an equation or a variable, has exactly two adjacent edges. Thus, the matching choice is binary. Since we have built a simple cycle for each literal, there are only two matching solutions for each cycle, one selecting $l_{\ell,1}$ and all the odd numbered arcs, the other one selecting all the even numbered ones.

Additionally, due to how we constructed each cycle, the arc connecting the nodes belonging to a clause node set $N_c$ for cases (a) and (b) of \cref{fig:butterflies} is odd, while in case (c) it is even. Thus if the arc in the red box of \cref{fig:butterflies} (a) is selected, then all the arcs in boxes of subgraph type (b) will also be selected, the arcs in boxes of subgraph type (c) will not be selected, and \emph{vice versa}.

It follows that an arc pair inside a labelled node set can be selected if and only if the literal assignment -- as read from the graph -- satisfies the corresponding clause. As such, the objective functions of \mtwosat and \vam, given the proposed graph construction and interpretation, are equivalent. This implies that \vam is NP-complete.
\end{proof}

\paragraph{Example}
For the convenience of the reader, we complement the formal proof with an example of how the intermediate bipartite graph can be built from the following list of AND clauses:
\begin{equation}
\label{eq:red-example}
\{a \myand b,\; \mynot{a} \myand c,\; c \myand d\}
\end{equation}

\begin{figure}
\centering
\includegraphics[width=\textwidth]{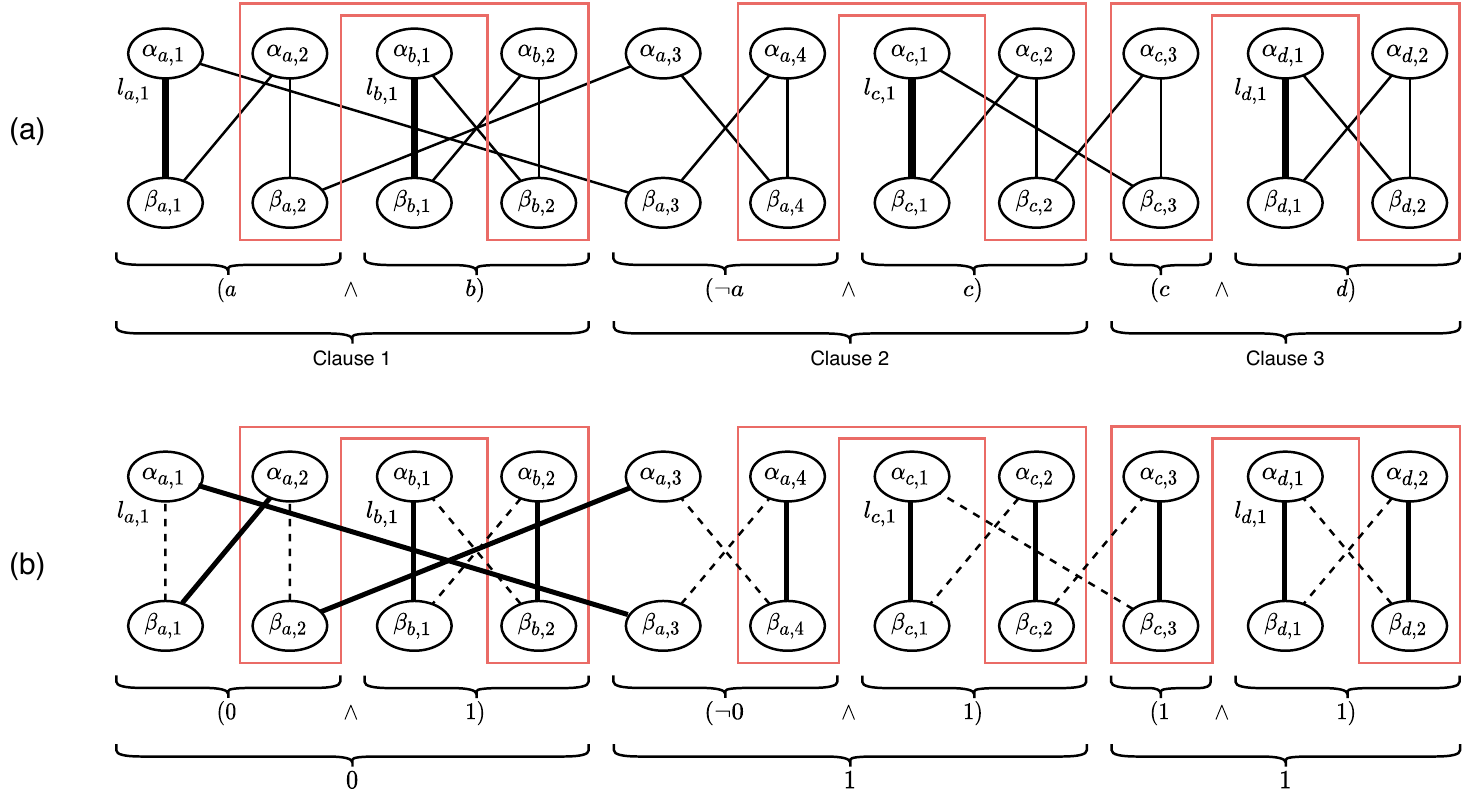}
\caption{\label{fig:butterflies-example}Bipartite graphs representing
the list of AND literals shown in \cref{eq:red-example}.
Graph (a) shows the correspondence between each clause and literal
and the nodes and arcs in the graph.
Graph (b) shows an optimal matching for this graph. Solid lines
represent matched arcs, dashed lines represent non-matched arcs.}
\end{figure}

The list of clauses is scanned left to right. The first clause encountered
is $(a \myand b)$, and both $a$ and $b$ are literals being encountered
for the first time. As a result, in the graph we add two structures of
type (a) shown in \cref{fig:butterflies}.
The start node of $a$ is $\alpha_{a,1}$, the end node of $a$ is $\beta_{a,2}$.
In the same way,
the start node of $b$ is $\alpha_{b,1}$, the end node of $b$ is $\beta_{b,2}$.

Now, the second clause $(\mynot{a} \myand c)$ is processed.
The literal $a$ was already encountered, and appears in negated form;
as a result, a structure of type (c) from \cref{fig:butterflies} is
inserted, and the $\beta_{a,2}$ end node is connected with $\alpha_{a,3}$.
The end node of $a$ is changed to the newly inserted $\beta_{a,3}$.
Instead, the literal $c$ is newly encountered, and therefore we insert
a structure of type (a) from \cref{fig:butterflies}.
The start node of $c$ is $\alpha_{c,1}$, and the end node of $c$ is $\beta_{c,2}$.

The process has now arrived at the third and last clause $(c \myand d)$.
The literal $c$ makes a reappearance in positive form, and therefore a
structure of type (b) is inserted.
The end node $\beta_{c,2}$ is connected to the new node $\alpha_{c,3}$, and
$\beta_{c,3}$ becomes the new end node of $c$.
Finally, the literal $d$ is encountered for the first time, and a structure of type (a)
is inserted. At this point the list of clauses is exhausted.

The last step for constructing the scalar graph consists in connecting the end
node of each literal with its corresponding start node.
For literal $a$, we connect $\alpha_{a,1}$ with $\beta_{a,3}$.
For literal $b$, we connect $\alpha_{b,1}$ with $\beta_{b,2}$.
For literal $c$, we connect $\alpha_{c,1}$ with $\beta_{c,3}$.
Finally, for literal $d$, we connect $\alpha_{d,1}$ with $\beta_{d,2}$.

\Cref{fig:butterflies-example} (a) shows the bipartite graph generated through
the steps we have just outlined. Additionally, the figure outlines in red each
clause node, and highlights the first arc of each cycle, which is used to
determine the value of each literal from the matching.

\Cref{fig:butterflies-example} (b) highlights the optimal matching
of the same graph, alongside with the value of each clause and each literal.
Since arc $l_{a,1}$ has not been used in the matching, $a$ is assigned value 0.
Instead, literals $b$, $c$ and $d$ are assigned a value of 1, 
since the arcs $l_{b,1}$, $l_{c,1}$, and $l_{d,1}$ are all used in the
matching.
This matching maximises the number of array equations and variables
matched, and therefore also minimizes $\Omega$: two arrays are matched, out of the three arrays in the input graph.
Each array corresponds to a clause, and the matched arrays represent
clauses whose value is one. In fact, the literal assignment makes
the second and third clause true, and the first clause false.


\section{Array-Aware Matching Algorithm}
\label{sec:algos}

The NP-completeness proof of optimal matching highlights how, in order to efficiently handle large-scale problems, there is a need to introduce suitable heuristics.
In this section, we propose an algorithm of reduced complexity.

Our proposal is a partially heuristic two-step procedure. The first step identifies obligatory matching choices and removes them from the problem, thereby reducing its size. We call this first step, presented in~\cref{sec:simplify-algo}, the \emph{simplification} step. No heuristics are involved in it.

For the second step, that we call the \emph{matching} step, we propose an extension of the Hopcroft-Karp~\cite{bib:HopcroftKarp-1973a} algorithm to array graphs, which allows to exploit local incidence matrices in such a way to preserve the existing looping constructs. The matching step, described in~\cref{sec:matching-algo}, terminates at the first solution found, whence its heuristic nature.

Before presenting the simplification and matching algorithms, it is necessary to define some operations on local incidence matrices that are used throughout the said algorithms. This is done in~\cref{sec:matrixops}, but requires an important preliminary remark.
When dealing with matching problems that contain multidimensional arrays, incidence matrices become multidimensional as well. To lighten the treatise, in this work we nonetheless stick to talking about ``rows'' and ``columns''. This notation does not cause any generality loss, however. All the proposed algorithms can work in the case of multidimensional incidence matrices by simply interpreting ``rows'' and ``columns'' as \emph{sets} of dimensions, referring respectively to equations and variables. We also talk about ``row'' and ``column vectors'', that generalize to matrices having only one of the two sets of dimensions.

\subsection{Operations on local incidence matrices}
\label{sec:matrixops}

Since incidence matrices are boolean, it is trivial to define the logical operators \emph{conjunction} (\emph{and}, $\myand$), \emph{disjunction} (\emph{or}, $\myor$) as well as \emph{negation} (\emph{not}, $\neg$) on matrices of the same dimensions as the element-wise operations. Additionally, for convenience, we define the operation $a \setminus b$ (\emph{subtraction}) as $a \myand \neg b$.

We also define the $\myand$ and $\setminus$ operator where the first argument is a matrix and the second is a row vector. The behaviour in this case is equivalent to replacing the vector with a matrix where every element in each row is equal to the corresponding element in the vector. Equivalently, these operations are also defined with column vectors.

We further define the \emph{flattenRows} and \emph{flattenColumns} operations, each taking a matrix and producing, respectively, a column and row vector. Each element of these vectors is $1$ iff there is at least a $1$ in the corresponding row or column, respectively.

The last operation that we need to define is \emph{solveLocalMatchingProblem}, which given a local incidence matrix $u$ returns a list of possible match matrices, called \emph{match options}. More in detail, each of the returned options is a valid match, that satisfies the following properties:
\begin{enumerate}
\item $m_{i,j}=1 \implies u_{i,j}=1$,
\item Each option has at most one element at $1$ for each row and column.
\end{enumerate}
For example, applying \emph{solveLocalMatchingProblem} to matrix
\begin{equation}
u = \begin{bmatrix}
    1 & 0 & 0 \\
    1 & 1 & 0 \\
    1 & 0 & 1 \\
    \end{bmatrix}
\end{equation}
returns the set of possible match matrices
\begin{equation}
\left\{
    \begin{bmatrix}
    1 & 0 & 0 \\
    0 & 1 & 0 \\
    0 & 0 & 1 \\
    \end{bmatrix},\;
    \begin{bmatrix}
    0 & 0 & 0 \\
    0 & 0 & 0 \\
    1 & 0 & 0 \\
    \end{bmatrix},\;
    \begin{bmatrix}
    0 & 0 & 0 \\
    1 & 0 & 0 \\
    0 & 0 & 0 \\
    \end{bmatrix},\;
    \begin{bmatrix}
    1 & 0 & 0 \\
    0 & 0 & 0 \\
    0 & 0 & 0 \\
    \end{bmatrix}
    \right\}
\end{equation}
where, depending on the \emph{solveLocalMatchingProblem} implementation but without any effect on the proposed algorithms, the last match may or may not be returned as it can be considered part of the first (larger) one.
As an additional constraint, \emph{solveLocalMatchingProblem} shall try to return the largest possible match, that is, the one that matches the largest number of equation/variable pairs.
This is not meant as a strong requirement -- in other words, it is acceptable to provide an implementation of \emph{solveLocalMatchingProblem} that does not return the largest possible match options in all cases.
However, the largest the single options provided by this primitive, the higher quality the final matching will be.

\subsection{Simplification algorithm}
\label{sec:simplify-algo}

\begin{algorithm}[t]
    \Fn{simplify} {
        \textbf{Input}: $\underline{G}=(\underline{N} = \underline{V} \cup \underline{E}, \underline{D})$\;
        \textbf{Output}: $\underline{G}$\;
        $L \leftarrow \emptyset$\;
        \ForEach{$\underline{n} \in \underline{N}$} {
            \textbf{if} $\deg(\underline{n}) = 1$ \textbf{then} $L \leftarrow L \cup \{\underline{n}\}$\;
        }
        \While{$L \neq \emptyset$} {
            $\underline{n}_1 \leftarrow \mathrm{getElementIn}(L)$\;
            $L \leftarrow L \setminus \{\underline{n}_1\}$\;
            $\underline{n}_2 \leftarrow $ unmatched node reached from $\underline{n}_1$\;
            $matchOptions \leftarrow \mathrm{solveLocalMatchingProblem}(u_{12})$\;
            \If{$|matchOptions| = 1$} {
                $alreadyMatched1 \leftarrow$ vector of size $|\underline{n}_1|$ where $alreadyMatched1_j = 1$ iff $\underline{n}_{1,j}$ is matched\;
                $alreadyMatched2 \leftarrow$ vector of size $|\underline{n}_2|$ where $alreadyMatched2_j = 1$ iff $\underline{n}_{2,j}$ is matched\;
                $m_{12} \leftarrow \mathrm{getElementIn}(matchOptions) \setminus alreadyMatched1 \setminus alreadyMatched2$\;
                \If{all components of $\underline{n}_2$ are matched} {
                    $L \leftarrow L \setminus \{\underline{n}_2 \}$\;
                    \ForEach{$\underline{n} \; in \; unmatchedReachedNodes(\underline{n}_2)$} {
                        \textbf{if} $\mathrm{unmatchedDegree}(\underline{n}) = 1$ \textbf{then} $L \leftarrow L \cup \{\underline{n}\}$\;
                    }
                }
                \Else {
                    \textbf{if} $\mathrm{unmatchedDegree}(\underline{n}_2) = 1$ \textbf{then} $L \leftarrow L \cup \{\underline{n}_2\}$\;
                }
            }
        }
    }
    \caption[Array matching simplify algorithm]{Array matching simplify algorithm.}
    \label{alg:matching-simplify}
\end{algorithm}

The simplification step is dedicated to performing all the obligatory matches.
Its importance in the handling of real-world problems becomes evident when considering that, as explained in~\cite[Chapter 7, and more specifically Section 7.2]{bib:CellierKofman-2006a},
when matching differential equations, the \emph{state variables} of the system are considered to be known and need not be matched.
Thus, the variables to be matched are either non-state variables, or derivatives of state variables.
In modelling the evolution of physical systems over time, it is very common to write equations in the form
\begin{displaymath}
\dot{x_i} = \mathrm{f}(x,u)
\end{displaymath}
where $x$ and $u$ are respectively the sets of state variables and inputs.
These equations introduce equation nodes in the bipartite matching graph with a single arc connecting them to the corresponding derivative, and that thus can only be matched with one variable. For example, this applies to all the three arcs in~\cref{lst:forloopmodelica}.
The commonplace presence of obligatory matching options motivates the introduction of an efficient simplification step.

The proposed simplification step is shown in~\cref{alg:matching-simplify}. It takes as input the array graph $\underline{G}$, and operates as follows. First, it constructs a set $L$ from every node in the graph with only one arc, regardless of it being an equation or a variable.

Then, for every node $\underline{n}_1$ in the set, it uses the \emph{solveLocalMatchingProblem} procedure to attempt to match the only arc adjacent to that (array) node.

In the case of multiple matching options, the simplification step skips the considered node, because at this stage any matching choice would be arbitrary, and may affect the feasibility of the \vam problem.
Conversely, if only one option is found that that fully matches $\underline{n}_1$, that option is included in the solution.
Attention then shifts to $\underline{n}_2$, the only node reached by $\underline{n}_1$.

Since we have fully matched one node and matched some variables from $\underline{n}_2$, the simplify procedure now checks whether also $\underline{n}_2$ is fully matched. If it is, $\underline{n}_2$ is removed from the set in case it was there, and since having fully matched $\underline{n}_2$ may result in neighboring nodes with only one unmatched arc, all such nodes are added to the set.
Even if $\underline{n}_2$ is not fully matched, it too may end up having only one unmatched arc (if it previously had two), and in this case $\underline{n}_2$ is added to the set.
The simplify algorithm thus recursively eliminates all nodes with a constrained match, leaving only irreducible connected components, as well as nodes where arcs have multiple matching options. After the simplify step all matches found are removed from the graph, and only the remaining part of the graph is passed to the subsequent matching algorithm.

\subsubsection{Optimal matching in polynomial time}

It is evident from the formulation of the simplification step that
its algorithm executes in polynomial time with respect to the number of nodes
and arcs in the array graph and, with a suitable data structures to represent
incidence matrices, in O(1) time with respect to the size of the arrays.
We also notice that there is a class of array graphs that can be completely matched
by application of the simplification algorithm alone. The model of a wire previously
shown in \cref{lst:wire} is an example of that. By construction of
the simplification algorithm, such graphs must have a single solution to the
array-aware-matching problem, and that solution is therefore optimal.
It follows that such models can be optimally matched in polynomial time with
respect to the number of nodes and arcs in the array graph, and in O(1) time with
respect to the size of the arrays.

Finding other classes of graphs which can be optimally matched in polynomial time
is an open research issue.

\subsection{Matching algorithm}
\label{sec:matching-algo}

Now we illustrate the complete array-aware matching algorithm.
This algorithm must be able to match any valid array graph, and
additionally it must attempt to approximate the optimal matching
as much as possible.
The algorithm we present is based on the well-known Hopcroft-Karp
one for bipartite graph matching, with adaptations to support
array graphs. Its main procedure is shown in \cref{alg:matching-main}.

\begin{algorithm}[h!]
\Fn{matching} {
  \textbf{Input}: $\underline{G}=(\underline{N} = \underline{V} \cup \underline{E}, \underline{D})$\;
  \textbf{Output}: $\underline{G}$\;
  
  P = augmentingPaths($\underline{G}$)\;
  \While{$P \ne \emptyset$}{
    \ForEach{$p \in P$}{
      applyPath($G$, $p$)
    }
    $P$ = augmentingPaths($G$)\;
  }
}
\caption{Main procedure of the matching algorithm.}
\label{alg:matching-main}
\end{algorithm}

As we will see in the following, procedure \emph{augmentingPaths} computes a list of
non-intersecting \emph{augmenting paths} in the graph to be applied later
by the \emph{applyPath} procedure. When no augmenting paths are found,
the matching is considered complete.

\begin{algorithm}[h!]
\Fn{applyPath} {
  \textbf{Input}: $\vvv{G}, p$\;
  \textbf{Output}: $\vvv{G}$\;
  
  \ForEach{$(\vvv{n}, \vvv{d}_{i,j}, {m}'_{i,j}) \in p$}{
    \uIf{$\vvv{n}$ is the equation $\vvv{e}_i$}{
      ${m}_{i,j} \leftarrow {m}_{i,j} \myor {m}'_{i,j}$
    }\ElseIf{$\vvv{n}$ is the variable $\vvv{v}_j$}{
      ${m}_{i,j} \leftarrow {m}_{i,j} \setminus {m}'_{i,j}$
    }
  }
}
\caption{Procedure for adding an augmenting path to the matching.}
\label{alg:matching-apply}
\end{algorithm}

Each augmenting path $p$, when applied, adds or removes matching from
the graph.
It consists of a list of tuples called \emph{steps} $p_i = (\vvv{n}, \vvv{d}_{i,j}, m'_{i,j})$
where $\vvv{n} \in \vvv{N}$ is the starting node of the step, $\vvv{d}_{i,j}$ is the arc being
traversed, and $m'_{i,j}$ is the incidence matrix that specifies which elements
of the local matching incidence matrix ${m}_{i,j}$ are modified by the step.
In a similar way to the Hopcroft-Karp algorithm,
steps starting from an equation node ($\vvv{n} \in \vvv{E}$)
add non-zero matrix entries to the matching, and steps starting from a variable node
($\vvv{n} \in \vvv{V}$) remove entries from the matching.
In other words, in a step starting from an equation node,
the elements set to 1 in $m'_{i,j}$ are added to the matching matrix ${m}_{i,j}$.
On the contrary, in a step starting from a variable node, the elements set to 1 in $m'_{i,j}$
are \emph{removed} from ${m}_{i,j}$.

The computation of the augmenting paths is performed through a
breadth-first-search in the residual graph. 
In contrast to the conventional scalar matching process performed by the
Hopcroft-Karp algorithm, when computing the augmenting
path we also need to keep track
of the matching matrices that express the set of equivalent scalar variables
and equations that are being matched or un-matched.

\begin{algorithm}[htb]
\Fn{augmentingPaths} {
  \textbf{Input}: $\vvv{G}$\;
  \textbf{Output}: $P$\;
  $\;$\\
  (Calculation of the initial frontier)\\
  $F \leftarrow \{\}$\;
  \ForEach{$\vvv{e}_i$}{
    $f \leftarrow$ vector of size $|\vvv{e}_i|$ where $f_j = 1$ iff $\vvv{e}_{i,j}$ is not currently matched\;
    \If{$f \ne \emptyset$}{
      $F \leftarrow$ append($F$, $\{(\vvv{e}_i, f)\}$)\; 
    }
  }
  $\;$\\
  (Calculation of the augmenting paths with breadth-first-search)\\
  $G_{BFS} = (A, B), L \leftarrow \mathrm{bfs}(G, F)$\;
  $\;$\\
  (Heuristic sort of the augmenting paths)\\
  $L \leftarrow \mathrm{heuristicSort}(L)\;$\\
  $\;$\\
  (Restriction of the flow of each augmenting path and removal of overlapping paths)\\
  $P \leftarrow \emptyset$\;
  \ForEach{$l_0 = (\vvv{n}, s_0) \in L$}{
    $p \leftarrow \{\}$,\,
    $l \leftarrow l_0$,\,
    $s \leftarrow s_0$\;
    \While{$\exists (l', m, l) \in B$}{
      $l \leftarrow l'$,\,
      $m' \leftarrow s \myand m$\;
      \uIf{$\vvv{n}', \vvv{n}$ are equations and variable $\vvv{e}_i, \vvv{v}_j$}{
        $s \leftarrow \textrm{flattenRows}(m')$\;
          $p \leftarrow \mathrm{append}(p, \{(\vvv{e}_i, \vvv{d}_{i,j}, m')\})$
      }
      \ElseIf{$\vvv{n}', \vvv{n}$ are variable and equation $\vvv{v}_j, \vvv{e}_i$}{
        $s \leftarrow \textrm{flattenColumns}(m')$\;
          $p \leftarrow \mathrm{append}(p, \{(\vvv{v}_j, \vvv{d}_{i,j}, m')\})$
      }
    }
    $\;$\\
    (Scalar arc intersection test between augmenting paths)\\
    \If{%
      $\nexists%
        \vvv{e}_i \in \vvv{E}, p' \in P:%
        (\vvv{e}_i, \vvv{d}_{i,j}, m) \in p, (\vvv{e}_i, \vvv{d}_{i,j}, m') \in p',%
        \textrm{flattenRows}(m) \myand \textrm{flattenRows}(m') \neq \emptyset$%
        \\
       \scalebox{0.98}[1]{$\myand \nexists %
         \vvv{v}_j \in \vvv{V}, p' \in P: %
         (\vvv{v}_j, \vvv{d}_{i,j}, m) \in p, (\vvv{v}_j, \vvv{d}_{i,j}, m') \in p', %
         \textrm{flattenColumns}(m) \myand \textrm{flattenColumns}(m') \neq \emptyset$}%
    }{
        $P \leftarrow \mathrm{append}(P, \{p\})$\;
    }
  }
}
\caption{Procedure for computing an augmenting path from a partial state of the matching.}
\label{alg:matching-makepaths}
\end{algorithm}

The breadth-first-search procedure must first be seeded with an initial frontier $F$
by collecting the list of array equation nodes with at least one free scalar equation.
This initial frontier is then passed to the \emph{bfs} procedure, which computes
both the initial list of augmenting path candidates $L$ and the forest of search trees
$G_{BFS}$ traversed during the search. In $G_{BFS}$, $A$ is the list of nodes in the
forest, and $B$ is the list of arcs.
Each node $a\in A$ is a tuple $(n, v)$ where $n$ is a node in the matching graph, and
$v$ is a binary vector with size and dimensionality equivalent to the one of $n$
specifying which scalar equations are being traversed in the path.
Similarly, each arc $b \in B$ is a triple $(a, m, a')$ where $a$ is the parent node,
$a'$ the child node, and $m$ is the local incidence matrix that describes the equivalent traversed
path in the homomorphic scalar graph.
The augmenting path candidates $l \in L$ are actually just leaves of the $G_{BFS}$ forest.

In contrast to the basic Hopcroft-Karp matching algorithm, but similarly to the
Ford-Fulkerson flow maximization algorithm~\cite{bib:FordFulkerson-1956a},
the initial ``flow'' in the first
steps of the path can be different from the final ``flow'' at the last step.
In this context the ``flow'' of a given step is simply the number of scalar equations affected.
To perform this operation, each path is traversed backwards -- from the end to the beginning --
and it is modified such that the following invariant is respected:
\begin{equation*}
\begin{array}{ll}
\mathrm{flattenRows}(m) = \mathrm{flattenRows}(m') & \forall\; (a, m, a'), (a', m', a'') \in A : \text{$a'$ equation node}\\
\mathrm{flattenColumns}(m) = \mathrm{flattenColumns}(m') & \forall\; (a, m, a'), (a', m', a'') \in A : \text{$a'$ variable node}\\
\end{array}
\end{equation*}
During this process, we build the augmenting path $p$ from the nodes, arcs, and matching matrices traversed.

Additionally,
the breadth-first-search does not immediately return a set of paths that respect the
\emph{non-intersection condition} already present in the Hopcroft-Karp algorithm.
An augmenting path intersects another if
the two paths, in at least one point, traverse the same node with intersecting
matching matrices.
In order to guarantee this property, each augmenting path is tested against
the others.
If two paths are intersecting, one of the two is discarded.

Discarding intersecting paths has the effect of eliminating multiple candidates which are
in mutual exclusion between each other.
The specific candidates being discarded at each step influence the solution
$\Omega$, and thus how close the solution is to the optimum.
Additionally, they impact the number of steps required by the matching algorithm.
In our current implementation, such selection depends on the ordering of $L$.
Our heuristic (implemented in the \emph{heuristicSort} procedure)
sorts $L$ based on the number of ones in the matching matrices of the path
(paths with more ones are prioritized).
Other heuristics could be devised to improve the solution $\Omega$ and the number of
steps required to complete the matching, and this could be an interesting direction for future work.

\begin{algorithm}[htb]
\Fn{bfs} {
  \textbf{Input}: $\vvv{G}, F$\;
  \textbf{Output}: $G_{BFS} = (A, B), L$\;

  $F' \leftarrow \{\}$\;
  $L \leftarrow \{\}$\;
  \While{$F \neq \{\} \myand P = \{\}$}{
    \ForEach{$a = (\vvv{n}, f) \in F$}{
      \ForEach{$\vvv{d}_{i,j}$ adjacent to $\vvv{n}$}{
        \uIf(move from equation to variable){$\vvv{n}$ is the equation $\vvv{e}_i$}{
          $S \leftarrow \mathrm{solveLocalMatchingProblem}((u_{i,j} \setminus m_{i,j}) \myand f)$\;
          \ForEach{$s \in S$}{
              $t \leftarrow$ vertical vector of size $|\vvv{v}_j|$ where $f_k = 1$ iff $\vvv{v}_{j,k}$ is not currently matched\;
              $m \leftarrow s \myand t$\;
              \uIf{$m \ne \emptyset$}{
                $a' = (\vvv{e}_i, \mathrm{flattenRows}(m))$\;
                $A \leftarrow A \union \{a'\}$,\,
                $B \leftarrow B \union \{(a, m, a')\}$,\,
                $L \leftarrow \mathrm{append}(L, \{a'\})$\;
              }\Else{
                $a' = (\vvv{e}_i, \mathrm{flattenRows}(s))$\;
                $A \leftarrow A \union \{a'\}$,\,
                $B \leftarrow B \union \{(a, s, a')\}$,\,
                $F' \leftarrow \mathrm{append}(F', \{a'\})$\;
              } 
          }
        }
        \ElseIf(move from variable to equation){$\vvv{n}$ is the variable $\vvv{v}_j$}{
          $S \leftarrow \mathrm{solveLocalMatchingProblem}(m_{i,j} \myand f)$\;
          \ForEach{$s \in S$}{
              $a' = (\vvv{v}_j, \mathrm{flattenColumns}(s))$\;
              $A \leftarrow A \union \{a'\}$,\,
              $B \leftarrow B \union \{(a, s, a')\}$,\,
              $F' \leftarrow \mathrm{append}(F', \{a'\})$\;
          }
        }
      }      
    }
    $F \leftarrow F'$,\,
    $F' \leftarrow \{\}$\;
  }
}
\caption{Breadth-first-search of augmenting paths in the matching graph.}
\label{alg:matching-bfs}
\end{algorithm}

The breadth-first-search procedure is shown in \cref{alg:matching-apply}.
It operates in the conventional way, with three additional constraints:
\begin{enumerate}
\item Each move in the search is associated with a \emph{vector of tangent elements} ${f}$
to the destination node $\vvv{n}$ of the move. $f$ is a vertical vector if $\vvv{n}$ is an
equation, otherwise $f$ is an horizontal vector.
\item Each move in the search must be associated with an incidence matrix $s$ called \emph{path matrix} that
satisfies the same conditions imposed on matching matrices $m_{i,j}$ (see \cref{sec:arrayawarematching},
conditions (3--6)).
\item The incidence matrix $s$ only affects the scalar variables or equations specified by $f$.
\end{enumerate}
These constraints ensure that each path traversed during the search corresponds
to a set of one or more equivalent paths in the scalar graph.
As a result, multiple different moves can start at the same moment from the
same node and through the same arc, but with different path matrices.
The computation of the set of possible moves at each iteration
is performed by enumerating the adjacent edges to the current
node, and then by using \emph{solveLocalMatchingProblem} to compute
the set of distinct valid ways to traverse that edge.
When a path reaches a variable node, and the path matrix contains at least one
non-zero column not corresponding to any matched scalar variable,
then it allows to augment the matching and the search is stopped.
Since the search process is stopped for the entire frontier,
all paths returned by the \emph{bfs} procedure have the same length in
terms of number of steps.

\subsection{Data representation to achieve O(1) complexity}
\label{sec:gmis-vaf}

To guarantee constant-time scaling with the size of array variables and equations,
every operation performed on incidence matrices described in \cref{sec:matrixops}
can be completed in constant time with respect to the size of the matrices themselves.

However, it can be easily recognised that in reality this is not possible without at least setting an upper bound on the size of the matrices, and setting such a bound would limit the applicability of our approach to large equation systems (which is precisely our goal).

As a consequence, while multidimensional incidence matrices and vectors
proved useful as a conceptual tool for explaining our methodology, as a
data structure they are not suitable \emph{as is} for implementation.

In order to achieve $\mathrm{O}(1)$ scaling we thus introduce two data
structures to replace multidimensional incidence matrices and vectors,
named \emph{Multidimensional Compressed Index Set} (MCIS) and
\emph{Multidimensional Compressed Index Map} (MCIM).
The first data structure, the MCIS, replaces multidimensional vectors,
and the second (MCIM) replaces multidimensional incidence matrices.

These two data structures are able to represent the entire range
of possible vectors and matrices that can appear in a matching problem,
but the operations on them are not constant-time in general.
However, the operations are largely constant-time wherever the arrays
in a model come from the spatial discretisation of partial-derivatives
differential equations, which cover virtually the totality of the
modelling cases of engineering interest.

Let us consider the set of multidimensional indices corresponding to the 1-elements
of a multidimensional vector.
Multidimensional Compressed Index Sets are data structures representing
sets of multidimensional indices as lists of \emph{multidimensional intervals}
or \emph{ranges}.

A \emph{multidimensional range} defined over field $\mathbb{K}=\mathbb{N}^n$ is a list of tuples $\{(a_1, b_1), (a_2, b_2), \dots (a_n, b_n)\}$, one tuple for each dimension in $\mathbb{K}$. It represents the set
\begin{equation*}
\{a_1, a_1+1, a_1+2, \dots b_1\} \times \{a_2, a_2+1, a_2+2, \dots b_2\} \times \dots
\end{equation*}
where $\times$ indicates the Cartesian product. For example, the multidimensional range $\{(1,3), (2,4)\}$ represents the following set of indices:
\begin{equation*}
\{(1,2),(2,2),(3,2),(1,3),(2,3),(3,3),(1,4),(2,4),(3,4)\}
\end{equation*}
In other words, a multidimensional range is defined as an hyperrectangle over field $\mathbb{K}$ represented with the coordinates of its vertices.

An MCIS is a list of multidimensional ranges, and it represents the union of the hyperrectangles represented in turn by each range in the list. Additionally, the ranges do not intersect.
A single multidimensional range of volume greater than 1 can also be represented as multiple adjacent ranges, and ranges may appear in any order.
As a result, the same index set can be represented with MCISes in multiple ways.

Let us now consider the set of multidimensional indices corresponding to
the 1-elements of a multidimensional incidence matrix.
Multidimensional Compressed Index Maps are data structures representing this
set of multidimensional indices, again as lists of intervals.

Each index in the set can be split in two parts: the sub-index corresponding
to the first set of dimensions $k$, and the sub-index corresponding to the second
set of dimensions $j$. For brevity we will represent each multidimensional index
in the set as $(k, j)$.
Now, to obtain an MCIM from a set of indices $A = \{(k_1, j_1), (k_2, j_2), \dots (k_n, j_n)\}$,
first we split $A$ in sub-sets $A_1, A_2, \dots A_m$
where each index $(k_i, j_i)$ satisfies the following identity:
\begin{equation}\label{eq:mcim-deriv-1}
\begin{array}{ll}
k_0 = \min_{1 < i \leq n}k_i\\
\delta = j_0 - k_0\\
j_i = k_i + \delta & \forall i: 1 < i \leq n\\
\end{array}
\end{equation}

In other words, the $j$ sub-indices in each subset $A_i$ must be such that they can be obtained just from the
corresponding $k$ sub-indices and the $\delta_i$, computed as described in \cref{eq:mcim-deriv-1}.
$\delta_i$ must be constant for all elements of each subset, but may be different between one
subset and another.
To make another comparison for the sake of explanation,
each subset $A_i$ represents a diagonal of an incidence matrix.

At this point, for each sub-set $A_i = \{(k_{i,1}, j_{i,1}), \dots (k_{i,n}, j_{i,n})\}$
consider the set $K_i = \{k_{i,0}, \dots k_{i,n}\}$.
We call a \emph{MCIM element} the tuple $\Sigma_i = (K_i, \delta_i)$, where $K_i$ is represented
as a Multidimensional Compressed Index Set.
Finally, an MCIM is the set of MCIM elements $\{\Sigma_1, \Sigma_2, \dots \Sigma_m\}$, each of
which corresponds to one of the sub-sets $A_i$.

The development of algorithms implementing the operations on the data structures
we have described is largely an engineering problem, and we do not
wish to delve into it.
Depending on the specific algorithms being chosen, the complexity of most operations
can range from an upper bound of $\mathrm{O}(n^2)$, to $\mathrm{O}(n)$ in the single-dimensional case
if the list of indices in each MCIS is kept ordered.
Better computational costs can be achieved by adopting well-known interval-tree
representations~\cite{bib:PreparataShamos-1985a}, which allow operations on MCISes to reach
$O(\log{n})$ complexity in some specific cases.

However, independently from the implementation we choose for operations on MCISes,
we can straightforwardly state that any operation on sets containing a single range
can be implemented in constant time.
The same holds for MCIMs with a single element $\Sigma_1$ containing a MCIS $K_1$
with a single range.
This will be the case for any model where the matching heuristic algorithm manages
to preserve all array equations and variables. Therefore, we expect that in such
cases -- which are typical in physical models~\cite{bib:AgostaEtAl-2019b} -- the
compilation process will be performed in constant time with respect to the size
of array equations and variables.

The only operation that can be performed in constant $\mathrm{O}(1)$ time on any
arbitrary MCIM is \emph{solveLocalMatchingProblem}.
In fact, it is trivial to prove that each MCIM element $\Sigma_i$
represents a matrix that satisfies the
properties outlined in \cref{sec:matrixops} for valid matrices returned
by \emph{solveLocalMatchingProblem}.
Therefore, \emph{solveLocalMatchingProblem} can be elided from the simplify
and matching algorithms, and the iteration on the matching options can
be replaced with the iteration of each $\Sigma_i$ in the input MCIM.


\section{Conclusions and future work}
\label{sec:conclusions}

We discussed the problem of translating declarative EB models into imperative code, concentrating on the crucial step of equation/variable matching. Relating our work to the research \emph{scenario}, we evidenced two issues to address. First, currently established approaches to EB-to-imperative translation handle array variables and equation looping constructs in an extremely inefficient manner. Second, array-aware proposals in the literature aim for a fast translation, but not for obtaining an efficient imperative code.

We showed that to pursue both the objectives above, one needs to translate an EB model in such a way to maximise the preservation of looping constructs in the imperative code, and we introduced a metric to measure the said preservation. This led us to define the concept of \emph{optimal} array-aware matching, as the one that maximises that metric.

As our main methodological contribution, we proved that the problem of computing an optimal matching in the sense just stated is NP-complete. Motivated by this completeness, we proposed an algorithm to compute an array-aware matching in polynomial time, whose \emph{rationale} is to stop at the first solution found, privileging however the preservation of looping constructs when choices need to be taken.

The ideas we presented are currently being put to work within the implementation of an experimental Modelica compiler~\cite{bib:AgostaEtAl-2019a,bib:AgostaEtAl-2019b}. Indeed, based also on the effort that such a realisation entails, we do hope that in the next years the developments we described will be adopted in the EB modelling community at large. In fact, the advantages of an array-aware EB-to-imperative translation are essential for addressing large-scale models in an industrial context.

In the future, we also plan to extend our approach to the rest of the translation pipeline, addressing other problems such as equation scheduling and SCC resolution. Additionally, a more precise characterization of the set of array graphs that can be optimally matched in polynomial time is also of great interest, as well as better heuristics for the \vam algorithm.

\section{Acknowledgements}
The authors would like to thank the anonymous reviewers for their helpful suggestions that allowed us to improve the paper presentation, in particular regarding the definition of array equation and its impact on the optimality metric, as well as the opportunity to perform optimal matching in polynomial time for certain classes of models.

\bibliographystyle{ACM-Reference-Format}
\bibliography{TOMS-2023}

\end{document}